\documentclass[a4paper, oneside, reqno]{amsart}
\usepackage{amssymb,amsmath, mathtools}
\usepackage{graphicx}
\usepackage{epstopdf}
\usepackage{color}
\usepackage[margin=0cm]{caption}
\usepackage{lscape}

\def \R{\mathbb{R}}
\def \N{\mathbb{N}}
\def \Q{\mathbb{Q}}

\newtheorem{theorem}{Theorem}[section]

\begin{document}

\title[The Capacity of a Discrete Noisy Channel]{\textbf{On the computation of the Shannon capacity of a discrete channel with noise}}
\author{\textbf{Simon Robin Cowell}}
\thanks{This research was funded by the European Union through the European Regional Development Fund under the Competitiveness Operational Program (\emph{BioCell-NanoART} = \emph{Novel Bio-inspired Cellular Nano-architectures}, POC-A1-A1.1.4-E nr. 30/2016).}
\address{College of Science, Department of Mathematical Sciences\\ 
United Arab Emirates University\\ 
Al Ain, P.O.Box 15551, United Arab Emirates\\
E-mail: scowell@uaeu.ac.ae}
\date{}
\begin{abstract}
Muroga \cite{Muroga} showed how to express the Shannon channel capacity of a discrete channel with noise \cite{Shannon} as an explicit function of the transition probabilities. His method accommodates channels with any finite number of input symbols, any finite number of output symbols and any transition probability matrix. Silverman \cite{Silverman} carried out Muroga's method in the special case of a binary channel (and went on to analyse ``cascades" of several such binary channels).

This article is a note on the resulting formula for the capacity $C(a,c)$ of a single binary channel. We aim to clarify some of the arguments and correct a small error. In service of this aim, we first formulate several of Shannon's definitions and proofs in terms of discrete measure-theoretic probability theory. We provide an alternate proof to Silverman's, of the feasibility of the optimal input distribution for a binary channel. For convenience, we also express $C(a,c)$ in a single expression explicitly dependent on $a$ and $c$ only, which Silverman stopped short of doing.
\end{abstract}
\maketitle

\section{Introduction}

We recommend the beautifully written \cite{Shannon} to the reader wanting to understand the information theory discussed in the present paper. We begin by recalling a few definitions and theorems from that book.

In section 6 of chapter I of \cite{Shannon}, Shannon represents a discrete source of information by a discrete random variable $X$, taking values in $\{x_1, x_2, \ldots, x_n\}$ with probabilities $\{p_1, p_2, \ldots, p_n\}$, respectively. He proposes to find a way to measure the amount of choice, uncertainty, information or ``entropy" involved in a single sampling of $X$. This measurement should be a function $H(p_1, p_2, \ldots, p_n)$ of the distribution of $X$, and, Shannon suggests, should obey a certain set of 3 axioms. We begin by reformulating his definitions and axioms in terms of modern i.e. measure-theoretic probability theory (albeit that we use only the discrete measure).

From here on, unless stated otherwise, all probability spaces will be assumed to be equipped with their discrete $\sigma$-algebra. For example, in the probability space $(\Omega, \Sigma, P)$, $\Sigma$ will simply be $\mathcal{P}(\Omega)$, the power set of $\Omega$. Therefore we will supress the notation $\Sigma$, and in place of $(\Omega, \Sigma, P)$ we will simply write $(\Omega, P)$.

We want $H$ to be a function associating with every finite, discrete probability space $(\Omega, P)$ a non-negative real number $H(\Omega, P)$, and we want $H$ to obey certain axioms. We will state those axoims, and then, closely following Shannon's proof sketch, we will prove that such an $H$ exists, and is unique, up to a multiplicative constant. One of our axioms will be that $H$ must be invariant with respect to probability-preserving bijections between the possible outcomes (i.e. those having strictly positive probability) of finite, discrete probability spaces. In other words, $H$ will depend only on the multiset of probabilities of all outcomes in the space having positive probability. The reader might ask why then, we don't describe $H$ only as a function of $n \in \N$ and $\vec{P} \in \Delta^{n-1}$ which is invariant with respect to permutations of the coordinates of $\vec{P}$, without introducing the redundant measure space ? The answer is just that we find the measure space convenient in reformulating Shannon's 3rd axiom precisely (as our 4th axiom), see below.
We also find the 4th axiom below to be convenient as it yields immediately the fact that $H(X,Y)=H(X)+H(Y|X)$, see section 2, below.

Let $(\Omega, P)$ be a finite, discrete probability space with $|\Omega| = n \in \N$ and label the outcomes, that is, the elements of $\Omega$, as $\Omega = \{\omega_1, \omega_2, \ldots, \omega_n\}$, so that the vector of probabilities $\vec{P} = (P(\{\omega_1\}), P(\{\omega_2\}), \ldots, P(\{\omega_n\}))$ takes values in $\Delta^{n-1}$, the standard $(n-1)$-simplex.

Let $\Omega_+$ denote the set $\{\omega \in \Omega : P(\{\omega\})>0\}$ of outcomes having strictly positive probability.

Whenever we have a partition $\mathcal{E}$ of $\Omega$, let us define a probability measure $P_{\mathcal{E}}$ on the discrete $\sigma$-algebra on $\mathcal{E}$ by
\begin{equation}
P_{\mathcal{E}}(\{E_i \in \mathcal{E} : i \in I\}) = P \left ( \bigsqcup_{i \in I} E_i \right ) = \sum_{i \in I} P(E_i),
\end{equation}
where $I$ is any index set, and $\sqcup$ denotes disjoint union. Thus $(\mathcal{E}, P_{\mathcal{E}})$ is a discrete probability space.
Also, for each $E \in \mathcal{E}$ such that $P(E) > 0$, and for each subset $F \subseteq E$, let us denote by $Q_E(F)$ the conditional probability
\begin{equation}
Q_E(F) = P(F|E),
\end{equation}
so that $(E, Q_E)$ is a discrete probability space.

We are now ready to reformulate Shannon's axioms in these terms:
\begin{enumerate}
\item
Whenever $(\Omega, P)$ and $(\Omega', P')$ are finite, discrete probability spaces and there is a bijection $f: \Omega_+ \to \Omega_+'$ such that $P'(\{f(\omega)\}) = P({\omega})$ for all $\omega \in \Omega_+$, we must have that $H(\Omega, P) = H(\Omega', P')$.
\item
$H(\Omega, P)$ must be continuous in the probability vector $\vec{P} \in \Delta^{n-1}$ with respect to the topology which $\Delta^{n-1}$ inherits from $\R^n$.
\item
For all $n \in \N$, for any discrete probability spaces $(\Omega, P)$ and $(\Omega', P')$ such that $|\Omega| = n$, $|\Omega'| = n+1$, $P(\{\omega\}) = \frac{1}{n}$ for all $\omega \in \Omega$ and $P'(\{\omega'\}) = \frac{1}{n+1}$ for all $\omega' \in \Omega'$, we must have
\begin{equation}
H(\Omega, P) < H(\Omega', P').
\end{equation}
\item
For any partition $\mathcal{E}$ of $\Omega$, we must have that
\begin{equation}
\label{eq: additivity axiom}
H(\Omega, P) = H(\mathcal{E}, P_{\mathcal{E}}) + \sum_{\mathclap{\substack{E \in \mathcal{E} \\ P(E)>0}}} P(E) H(E, Q_E).
\end{equation}
\end{enumerate}
We claim that the first and second axioms are natural.
Note that it follows from the first axiom that, if it is convenient for the computation of $H$, we may delete from $\Omega$ any outcomes $\omega$ having zero probability. Such outcomes exist precisely when $n \geq 2$ and $\vec{P}$ belongs to the boundary of $\Delta^{n-1}$. Deleting $k$ outcomes having zero probability in effect replaces $\Delta^{n-1}$ by a copy of $\Delta^{n-k-1}$ which is isomorphic to the part of the boundary of $\Delta^{n-1}$ in question. In the extreme case $k = n-1$, we reduce $\Delta^{n-1}$ to a copy of $\Delta^0$, that is, the singleton set $\{1\}$.
The idea of the third axiom is that, if all outcomes are equally likely, then the amount of choice, or uncertainty, should be greater, when there are more possible outcomes.
In Shannon's words, the idea of the fourth axiom is that ``If a choice be broken down into two successive choices, the original $H$ should be the weighted sum of the individual values of $H$." Shannon illustrates his meaning with an example (fig 6 of section 6 of chapter 1 in \cite{Shannon}), of which he writes:
``At the left we have three possibilities with probabilities $p_1=\frac{1}{2}$, $p_2 = \frac{1}{3}$ and $p_3 = \frac{1}{6}$. On the right we first choose between two possibilities, each with probability $\frac{1}{2}$, and if the second occurs make another choice between two possibilities with probabilities $\frac{2}{3}$ and $\frac{1}{3}$. The final results have the same probabilities as before. We require, in this special case, that
\[
H \left ( \tfrac{1}{2}, \tfrac{1}{3}, \tfrac{1}{6} \right ) = H \left ( \tfrac{1}{2}, \tfrac{1}{2} \right ) + \tfrac{1}{2} H \left ( \tfrac{2}{3}, \tfrac{1}{3} \right ).
\]
The coefficient $\frac{1}{2}$ on the right-hand side is because this second choice only occurs half the time."

Translating this example to our terminology, the partition $\mathcal{E}$ has two parts, $E_1 = \{x_1\}$ and $E_2 = \{x_2, x_3\}$.
Our equation \eqref{eq: additivity axiom} becomes
\begin{align}
\notag
H(\{x_1, x_2, x_3\}, (\tfrac{1}{2}, \tfrac{1}{3}, \tfrac{1}{6}))
&= H(\{E_1, E_2\}, (\tfrac{1}{2},\tfrac{1}{2})) \\
&\phantom{potatoes} + \left [ P(E_1) H(E_1, (1)) + P(E_2) H(E_2, (\tfrac{2}{3}, \tfrac{1}{3})) \right ] \\
&= H(\{E_1, E_2\}, (\tfrac{1}{2},\tfrac{1}{2})) + \left [ \tfrac{1}{2} H(E_1, (1)) + \tfrac{1}{2} H(E_2, (\tfrac{2}{3}, \tfrac{1}{3})) \right ],
\end{align}
where we represent the various probability functions $P$ by their corresponding probability vectors $\vec{P}$.
Later we will see that, in this example, the first term in the square brackets vanishes, because the entropy of the certain event is zero; no information is contained in an experiment whose outcome is known in advance.
We have introduced the condition $P(E)>0$ in axiom 4. One effect of this will be to introduce a similar condition in the formula for $H$ in Theorem \ref{thm: Existence and Uniqueness of Entropy}. This allows our probability vector $\vec{P}$ to remain in the standard simplex $\Delta^{n-1}$, whereas Shannon instead must delete any events with zero probability, reducing to a lesser $n$ and replacing the standard simplex $\Delta^{n-1}$ by a simplex of lesser dimension. Our solution is less elegant than Shannon's, but we like its comparative precision.

For simplicity, and without risk of confusion, we will write $P(\omega)$ for $P(\{\omega\})$.

\begin{theorem}[Existence and uniqueness of Entropy]
\label{thm: Existence and Uniqueness of Entropy}
Let $(\Omega, P)$ be a finite, discrete probability space.
Then the functions
\begin{equation}
H(\Omega, P) = -K \sum_{\mathclap{\substack{\omega \in \Omega \\ P(\omega)>0}}} P(\omega) \log_b P(\omega),
\end{equation}
where $K$ is a strictly positive constant and $b>1$ satisfy axioms 1 - 4. These are the only functions satisfying those axioms.
\end{theorem}
\begin{proof}
We reproduce Shannon's proof of his theorem, filling in some details. For each $n \in \N$ let $(\Omega_n, P_n)$ be a finite, discrete probability space with $|\Omega_n| = n$ and with $P_n$ the uniform probability measure on $\Omega_n$, that is, $P_n(\omega)=\frac{1}{n}$ for all $\omega \in \Omega_n$. Suppose for the sake of argument that a function $H$ exists which obeys axioms 1 - 4.
Let $A(n) = H(\Omega_n,P_n)$. Then by axiom (1), $A: \N \to \R$ is well-defined.
It follows from axiom (4) that
\begin{equation}
\label{eq: log-like property of A(n)}
A(st)=A(s)+A(t) \qquad \text{for all } s, t \in \N.
\end{equation}
We also have, by axiom (3), that $A(n)$ is strictly increasing in $n$.
Fix $s,t \in \N$ with $s,t \geq 2$, and let $n \in \N$.
Then provided $n$ is large enough, there exists a unique $m \in \N$ such that
\begin{equation}
\label{eq: inequalities of powers}
s^m \leq t^n < s^{m+1}.
\end{equation}
Since $A(n)$ is strictly increasing, we have
\begin{equation}
A(s^m) \leq A(t^n) < A(s^{m+1}),
\end{equation}
hence by \eqref{eq: log-like property of A(n)},
\begin{equation}
mA(s) \leq nA(t) < (m+1)A(s),
\end{equation}
and
\begin{equation}
\label{eq: inequality for the ratio of As}
\frac{m}{n} \leq \frac{A(t)}{A(s)} < \frac{m+1}{n}.
\end{equation}
In this last step we have used the fact that $A(s)$ must be positive.
Indeed, for any $t \in \N$ we have $A(t) = A(1\cdot t) = A(1)+A(t)$, hence $A(1) = 0$. But $A(n)$ is strictly increasing in $n$, therefore $A(n) > 0$ for all $n \geq 2$.
Let $b > 1$.
The function $\log_b (x)$ is also strictly increasing in $x$, therefore \eqref{eq: inequalities of powers} also implies that
\begin{equation}
\log_b s^m \leq \log_b t^n < \log_b s^{m+1},
\end{equation}
so
\begin{equation}
m \log_b s \leq n \log_b t < (m+1) \log_b s,
\end{equation}
and
\begin{equation}
\label{eq: inequality for the ratio of the logs}
\frac{m}{n} \leq \frac{\log_b t}{\log_b s} < \frac{m+1}{n}.
\end{equation}
Together, \eqref{eq: inequality for the ratio of As} and \eqref{eq: inequality for the ratio of the logs} imply that
\begin{equation}
\left |
\frac{A(t)}{A(s)} - \frac{\log_b(t)}{\log_b(s)}
\right |
< \frac{1}{n}.
\end{equation}
Since $n$ is arbitrary, we have
\begin{equation}
\frac{A(t)}{A(s)} = \frac{\log_b(t)}{\log_b(s)} \qquad \text{for all } s, t \in \N, s, t \geq 2.
\end{equation}
Fixing $s \in \N$ with $s \geq 2$, we have
\begin{equation}
A(t) = \frac{A(s)}{\log_b(s)} \log_b(t) \qquad \text{for all } t \in \N \text{ with } t \geq 2,
\end{equation}
hence
\begin{equation}
A(t) = K \log_b(t) \qquad \text{for all } t \in \N \text{ with } t \geq 2,
\end{equation}
where $K>0$ is a strictly positive constant depending on $b$.
Since $A(1)=0$, we even have
\begin{equation}
A(t) = K \log_b(t) \qquad \text{for all } t \in \N.
\end{equation}
So far we have found a formula for $H(\Omega,P)$ in the case that $P$ is the uniform distribution on $\Omega$, i.e. when all outcomes are equally likely.
We need to be able to relax this condition.
In fact, let $(\Omega, P)$ be a discrete probability space with $n$ outcomes, not necessarily equally likely, but having comensurable probabilities $P(\omega_i)$. Since the probabilities sum to $1$, this comensurability is equivalent to the $P(\omega_i)$ all being rational.
Assuming for simplicity that the $P(\omega_i)$ are all strictly positive, we can write
\begin{equation}
P(\omega_i) = \frac{s_i}{m} \qquad \text{for all } i,
\end{equation}
where $m \in \N$ satisfies $m > n$, and the $s_i \in \N$ satisfy $s_1 + \cdots + s_n = m$.
Now let $(\Omega', P')$ be a discrete probability space with $m$ equally likely outcomes, and let $\mathcal{E}$ be a partition of $\Omega$ into $n$ nonempty parts $E_1, \ldots, E_n$ with sizes $s_1, \ldots, s_n$, respectively.
By axiom (4) we have
\begin{equation}
A(m) = H(\Omega', P') = H(\mathcal{E}, P_{\mathcal{E}}) + \sum_i P(E_i) H(E_i, Q_{E_i}) = H(\Omega, P) + \sum_i P(\omega_i) A(s_i),
\end{equation}
hence
\begin{align}
H(\Omega,P)
&= A(m) - \sum_i P(\omega_i) A(s_i) \\
&= \sum_i P(\omega_i) \left [A(m) - A(s_i) \right ] \\
&= -K \sum_i P(\omega_i) \left [ \log_b(s_i) - \log_b(m) \right ] \\
&= -K \sum_i P(\omega_i) \log_b \left ( \frac{s_i}{m} \right ) \\
&= -K \sum_i P(\omega_i) \log_b P(\omega_i).
\end{align}
Note that in the case of equally likely outcomes, we recover $A(n)$ from this formula.
Now by axiom (2) and by the density of $\Q^n$ in $\R^n$ we can extend this formula even to the case of irrational probabilities, to obtain
\begin{equation}
H(\Omega, P) = -K \sum_i P(\omega_i) \log_b P(\omega_i)
\end{equation}
for any finite probability space whose outcomes have strictly positive probability.
Having thus shown that this form for $H$ is necessary, if the 4 axioms are to hold, we claim that it is also sufficient.
\end{proof}
Note that by changing our choice of base $b>1$, without any loss of generality we can omit the constant $K>0$, i.e. assume that $K=1$.
Indeed, if $b, b' > 1$ and $K > 0$, then $K \log_b x = K \frac{\log_{b'} x}{\log_{b'}b} = \frac{K}{\log_{b'}b} \log_{b'} x$, where $\frac{K}{\log_{b'}b}>0$ is also a positive constant.
With the convention that $K=1$, if we set $b=2$ then the units of entropy $H$ are known as ``bits", a contraction of ``binary digits", as explained in \cite{Shannon}. In fact, with $K=1$ and $b=2$, the entropy of a probability space having $2^n$ equally likely outcomes will be $A(2^n) = 1 \cdot \log_2 (2^n) = n$ bits. This makes sense, since $n$ binary digits can represent $2^n$ possible states.

Whenever $X$ is a random variable with finite range $\{x_1, x_2, \ldots, x_n\}$, we will write $H(X)$ to mean $H(\{x_1, x_2, \ldots, x_n\}, P)$, where $P$ is the discrete probability measure on the range of $X$ defined by $P(x_i) = P(X = x_i)$ for all $i \in \{1, 2, \ldots, n\}$.
So in the case of a random variable $X$ with finite range $\{x_1, x_2, \ldots, x_n\}$, the formula in Theorem \ref{thm: Existence and Uniqueness of Entropy} becomes
\begin{equation}
H(X) = -K \sum_{\mathclap{\substack{1\leq i \leq n \\ P(x_i)>0}}} P(x_i) \log P(x_i)
\end{equation}


\section{Muroga's explicit solution of Shannon's implicit equation for $C$}

Suppose that the input to a discrete channel is represented by a random variable $X$ with range $\{x_1, \ldots x_n\}$ and that the output is represented by a random variable $Y$ with range $\{y_1, \ldots y_m\}$.
Throughout this section we will assume that the message to be transmitted comprises a sequence of symbols being independent random samplings of $X$, and that the message is perturbed by noise in transmission, each symbol being perturbed independently.

Let $p_i = P(X = x_i)$, $r_j = P(Y = y_j)$ and $p_{i,j}=P(X=x_i \wedge Y=y_j)$.
Also let
\begin{equation}
q_{i,j}=
\left \{
\begin{array}{ll}
P(Y=y_j | X=x_i),	&	p_i \neq 0, \\
0,						&	p_i = 0.
\end{array}
\right .
\end{equation}
Note that
\begin{equation}
p_{i,j}=p_i q_{i,j},	\qquad \text{for all } i, j.
\end{equation}
Fix $b > 1$ to act as the base for any logarithms and exponentials. In case $b=2$ the units of entropy will be bits, which will of course be particularly appropriate when we study binary channels.
Also fix an arbitrary constant $w \in \R$.
Define
\begin{equation}
\log^* x =
\left \{
\begin{array}{ll}
\log_b x,	&	x > 0, \\
w,			&	x=0.
\end{array}
\right .
\end{equation}
Then the function $x \log^* x$ is continuous on the open interval $(0, \infty)$, and since the indeterminate limit $\lim_{x \to 0^+} x \log^* x = \lim_{x \to 0^+} x \log_b x$ exists and is equal to $0=0 \cdot w=0 \cdot \log^* 0$, it follows that $x \log^* x$ is actually continuous even on the closed interval $[0, \infty)$, and in particular on $[0,1]$.
The function $x \log^* x$ is differentiable on $(0, \infty)$, with $(x \log^* x)' = (x \log_b x)' = \log_b x + \frac{1}{\ln b} = \log_b x + \log_b e = \log^* x + \log^* e$ for all $x > 0$.
Hence $(x \log^* x)'$ tends to $-\infty$ as $x$ tends to $0$ from the right.
Also the right-hand derivative of $x \log^* x$ at $x=0$ is 
\begin{equation}
\lim_{h\to0^+} \frac{(0+h)\log^*(0+h)-0\cdot\log^* 0}{h} = \lim_{h\to0^+} \log^* h = \lim_{h\to0^+} \log_b h = -\infty.
\end{equation}
We have that
\begin{equation}
H(X)=-\sum_i p_i \log^* p_i,
\end{equation}
\begin{equation}
H(Y)=-\sum_j r_j \log^* r_j,
\end{equation}
and
\begin{equation}
H(X,Y)=-\sum_{i,j} p_{i,j} \log^* p_{i,j},
\end{equation}
where by $H(X,Y)$ we mean the entropy of the joint distribution of $X$ and $Y$.
We also define
\begin{equation}
H(Y|X=x_i)=-\sum_{j}q_{i,j}\log^* q_{i,j},
\end{equation}
the conditional entropy of $Y$ given that $X=x_i$. Note that in case $P(X=x_i)=0$, then the conditional probabilities $P(Y=y_j|X=x_i)$ are undefined, but by our definitions of $q_{i,j}$ and $\log^* x$, $H(Y|X=x_i)$ will in this case be equal to $0$.
We also define the expectation of this last defined quantity with respect to $X$ as the conditional entropy of $Y$ given $X$:
\begin{equation}
H(Y|X)= \sum_i p_i H(Y|X=x_i).
\end{equation}
From this definition, and from axioms (1) and (4) for $H$, it follows that
\begin{equation}
\label{eq: joint-conditional equation for Y given X}
H(X,Y)=H(X)+H(Y|X),
\end{equation}
and by the symmetry of the left-hand side, also that
\begin{equation}
\label{eq: joint-conditional equation for X given Y}
H(X,Y)=H(Y)+H(X|Y),
\end{equation}
We also have
\begin{align}
H(Y|X)
&=\sum_i p_i H(Y|X=x_i) \\
\label{eq: conditional entropy with qs}
&=-\sum_{i,j} p_i q_{i,j} \log^* q_{i,j} \\
&=-\sum_{i,j} p_{i,j} \log^* q_{i,j}.
\end{align}
\begin{theorem}
\label{thm: upper bound on entropy of product space}
If $X$ and $Y$ are random variables with finite range, then $H(X,Y) \leq H(X) + H(Y)$, with equality if, and only if, $X$ and $Y$ are independent.
\end{theorem}
\begin{proof}
\begin{align*}
&H(X)+H(Y)-H(X,Y) \\
&= -\sum_i p_i \log^* p_i - \sum_j r_j \log^* r_j - \left (- \sum_{i,j} p_{i,j} \log^* p_{i,j} \right ) \\
&= -\sum_i \left (\sum_j p_{i,j} \right ) \log^* p_i - \sum_j \left ( \sum_i p_{i,j} \right ) \log^* r_j - \left (- \sum_{i,j} p_{i,j} \log^* p_{i,j} \right ) \\
&= -\sum_{i,j} p_{i,j} \left (\log^* p_i + \log^* r_j - \log^* p_{i,j} \right ) \\
&= -\sum_{i,j:p_{i,j}>0} p_{i,j} \left (\log_b \frac{p_i r_j}{p_{i,j}} \right ).
\end{align*}
So, by the strict convexity of the logarithm,
\begin{align*}
&H(X)+H(Y)-H(X,Y) \\
&\geq -\log_b \sum_{i,j:p_{i,j}>0} \left (p_{i,j} \frac{p_i r_j}{p_{i,j}} \right ) \\
&= -\log_b \sum_{i,j:p_{i,j}>0} p_i r_j \\
&\geq -\log_b \sum_{i,j} p_i r_j \\
&=-\log_b \left [ \left ( \sum_i p_i \right ) \left ( \sum_j r_j \right ) \right ] \\
&= -\log_b [1 \cdot 1] \\
&= 0,
\end{align*}
with equality if and only if the following two conditions hold:
\begin{enumerate}
\item
$\frac{p_i r_j}{p_{i,j}} = c$, a constant, for all $i$ and $j$ such that $p_{i,j}>0$.
\item
For all $i,j$, $(p_i>0 \wedge r_j>0) \Rightarrow p_{i,j}>0$.
\end{enumerate}
Assuming that conditions (1) and (2) hold, from condition (1) we have that $p_i r_j = c p_{i,j}$ for all $i$ and $j$ such that $p_{i,j}>0$, hence
\begin{equation}
\sum_{i,j:p_{i,j}>0} p_i r_j = c \sum_{i,j:p_{i,j}>0} p_{i,j} = c \sum_{i,j} p_{i,j} = c.
\end{equation}
But it follows from condition (2) that $p_{i,j}>0$ if and only if both $p_i$ and $r_j$ are strictly positive, hence
\begin{equation}
c=\sum_{i,j:p_{i,j>0}} p_i r_j = \sum_{i,j:p_i>0 \wedge r_j>0} p_i r_j = \sum_{i,j} p_i r_j= \left ( \sum_i p_i \right ) \left ( \sum_j r_j \right ) = 1 \cdot 1 = 1.
\end{equation}
Therefore by condition (1) we have that $p_i r_j = p_{i,j}$ for all $i$ and $j$ such that $p_{i,j} > 0$, that is, by condition (2), for all $i$ and $j$ such that $p_i>0$ and $r_j>0$.
It follows that
\begin{equation}
p_{i,j} = p_i r_j \qquad \text{for all } i, j,
\end{equation}
and therefore, $X$ and $Y$ are independent.
Conversely, if we assume that $X$ and $Y$ are independent, then conditions (1) and (2) follow immediately.
\end{proof}
Now define the ``mutual information of $X$ and $Y$" as the deficit appearing in theorem \ref{thm: upper bound on entropy of product space}, i.e.
\begin{equation}
I(X,Y)=H(X)+H(Y)-H(X,Y).
\end{equation}
Then $I(X,Y)$ is non-negative. Moreover $I(X,Y)=0$ if, and only if, $X$ and $Y$ are independent.
By equations \eqref{eq: joint-conditional equation for Y given X} and \eqref{eq: joint-conditional equation for X given Y} we also have the relations
\begin{equation}
I(X,Y)=H(Y)-H(Y|X)
\end{equation}
and
\begin{equation}
I(X,Y)=H(X)-H(X|Y).
\end{equation}
The conditional entropy $H(X|Y)$ of the input $X$ given the output $Y$ is known as the ``equivocation".
Ideally, the equivocation would be zero, in which case $I(X,Y)$ would equal $H(X)$.
In case $H(X|Y)=H(X)$, that is, in case knowledge of the output $Y$ makes no difference in our degree of uncertainty as to the input $X$, we have that $I(X,Y) = 0$. Thus the mutual information $I(X,Y)$ measures the rate of transmission of information per symbol over the channel. Looked at another way, $I(X,Y)$ measures a type of correlation between $X$ and $Y$.

Shannon defines the ``capacity" $C$ of the channel in this case as the maximum mutual information over all possible distributions of the input $X$. That is,
\begin{equation}
C=\max_{\vec{p} \in \Delta^{n-1}} I(X,Y),
\end{equation}
where $\vec{p}=(p_1, \ldots, p_n)$.
Following Shannon and Muroga, we try to compute this maximum using the method of Lagrange multipliers. One of Muroga's many contributions is to account for the case in which the resulting maximising $\vec{p}$ is unfeasible, i.e. is outside $\Delta^{n-1}$. Muroga also accounts completely for the case that the transition matrix $Q = q_{i,j}$ is non-square, and even for the case where $Q$ is less than full rank. He shows exactly what must be done to compute $C$ in this general setting. Shannon seems to have overlooked these points.
By \eqref{eq: conditional entropy with qs} we have
\begin{align}
I(X,Y)
&=H(Y)-H(Y|X) \\
&=-\sum_j r_j \log^* r_j + \sum_{i,j} q_{i,j} p_i \log^* q_{i,j} \\
&=-\sum_j \left (\sum_i q_{i,j} p_i \right ) \log^* \left ( \sum_i q_{i,j} p_i \right ) + \sum_{i,j} q_{i,j} p_i \log^* q_{i,j} \\
&=-\sum_{i,j} q_{i,j} p_i \log^* \left ( \sum_i q_{i,j} p_i \right ) + \sum_{i,j} q_{i,j} p_i \log^* q_{i,j} \\
\label{eq: convenient form for optimal I}
&=\sum_{i,j} q_{i,j} p_i \left (\log^* q_{i,j} - \log^* \left ( \sum_i q_{i,j} p_i \right ) \right ),
\end{align}
where we point out that the index $i$ in the second summation is not bound by the first summation.

We are interested to know the partial derivative of $I(X,Y)$ at a fixed point $\vec{p} \in \Delta^{n-1}$, with respect to some particular $p_k$. For the time being we will suppose that $p_k \notin \{0,1\}$.
We have
\begin{align}
\frac{\partial}{\partial p_k} I(X,Y)
&= \frac{\partial}{\partial p_k} \left \{ \sum_{i,j} q_{i,j} p_i \left (\log^* q_{i,j} - \log^* \left ( \sum_i q_{i,j} p_i \right ) \right ) \right \} \\
&= \sum_{i,j} q_{i,j} \frac{\partial}{\partial p_k} \left \{ p_i \left (\log^* q_{i,j} - \log^* \left ( \sum_i q_{i,j} p_i \right ) \right ) \right \} \\
&= \sum_{i,j} q_{i,j} \left \{ \delta_{i,k} \log^* q_{k,j} - \frac{\partial}{\partial p_k} \left [ p_i \log^* \left ( \sum_i q_{i,j} p_i \right ) \right ] \right \} \\
\label{eq: note 1}
&= \sum_{i,j: r_j>0} q_{i,j} \left \{ \delta_{i,k} \log^* q_{k,j} - \frac{\partial}{\partial p_k} \left [ p_i \log_b \left ( \sum_i q_{i,j} p_i \right ) \right ] \right \} \\
&= \sum_{i,j: r_j>0} q_{i,j} \left \{ \delta_{i,k} \log^* q_{k,j} - \left [ \delta_{i,k} \log_b \left ( \sum_i q_{i,j} p_i \right ) + \frac{p_i}{\ln b} \frac{q_{k,j}}{\sum_i q_{i,j} p_i} \right ] \right \} \\
&= \sum_{j:r_j>0} q_{k,j} \log^* q_{k,j} - \sum_{j:r_j>0} q_{k,j} \log_b \left ( \sum_i q_{i,j} p_i \right ) \\
&\phantom{potatoespotatoespotatoes} - \frac{1}{\ln b} \sum_{j:r_j>0} \sum_i q_{i,j} p_i \frac{q_{k,j}}{\sum_i q_{i,j} p_i} \\
&= \sum_{j:r_j>0} q_{k,j} \log^* q_{k,j} - \sum_{j:r_j>0} q_{k,j} \log_b \left ( \sum_i q_{i,j} p_i \right ) - \frac{1}{\ln b} \sum_{j:r_j>0} q_{k,j} \\
&=\sum_j q_{k,j} \log^* q_{k,j} - \sum_j q_{k,j} \log^* \left ( \sum_i q_{i,j} p_i \right ) - \frac{1}{\ln b} \\
&=\sum_j q_{k,j} \left [ \log^* q_{k,j} - \log^* \left ( \sum_i q_{i,j} p_i \right ) \right ] - \frac{1}{\ln b},
\end{align}
where at \eqref{eq: note 1} we have used that fact that by definition, $q_{i,j}=0$ whenever $r_j=0$, whether or not $p_i = 0$, and also the fact that $\sum_i q_{i,j} p_i = r_j$.
Since $(1, \ldots, 1)$ is normal to $\Delta^{n-1}$, the method of Lagrange multipliers dictates that we want $\nabla I(X,Y)$ to be parallel to $(1, \ldots, 1)$.
It follows that we want
\begin{equation}
\label{eq: implicit equation for the p_i}
\sum_j q_{k,j} \left [ \log^* q_{k,j} - \log^* \left ( \sum_i q_{i,j} p_i \right ) \right ] = \mu \qquad \text{for all } k,
\end{equation}
where $\mu$ is some constant.
Multiplying by $p_k$ and summing over $k$,
\begin{equation}
\sum_{j,k} q_{k,j} p_k \left [ \log^* q_{k,j} - \log^* \left ( \sum_i q_{i,j} p_i \right ) \right ] = \mu \sum_k p_k = \mu.
\end{equation}
Given that the optimal values of $\vec{p}$ must obey this last equation (ignoring questions of feasibility for now), it follows from \eqref{eq: convenient form for optimal I} that
\begin{equation}
\mu = C
\end{equation}
(not $\mu = -C$, a mistake in Shannon which Muroga corrects).
We would like to isolate the $p_i$ in \eqref{eq: implicit equation for the p_i}.
We have that
\begin{equation}
\sum_j q_{k,j} \log^* \left (\sum_i q_{i,j} p_i \right ) = \sum_j q_{k,j} \log^* q_{k,j} - C \qquad \text{for all } k,
\end{equation}
or, as a matrix equation,
\begin{equation}
Q
\left (
\begin{array}{c}
\log^* \left ( \sum_i q_{i,1} p_i \right ) \\
\vdots \\
\log^* \left ( \sum_i q_{i,m} p_i \right )
\end{array}
\right )
=
\left (
\begin{array}{c}
\sum_j q_{1,j} \log^* q_{1,j} - C \\
\vdots \\
\sum_j q_{n,j} \log^* q_{n,j} - C
\end{array}
\right ),
\end{equation}
where $Q = q_{i,j}$ is the so-called transition matrix.
From this point, Shannon attempts to solve for $C$ in terms of $Q$ alone by inverting $Q$, in the special case in which $Q$ is square and invertible.
However, he is unable to eliminate the $p_i$ from his equation, meaning that his is an implicit rather than an explicit expression for $C$.
Muroga's analysis begins where Shannon left off, and successfully eliminates the $p_i$.
Here Shannon mistakenly calls $Q$ what is actually $Q^T$.
Muroga misses the opportunity to correct that mistake, instead simply writing that the order of suffices in a certain matrix product is different from the usual expression.
Although Muroga's conclusions are not harmed by this oversight, we take the chance to correct it here, and also to simplify slightly the argument which Muroga uses to eliminate the $p_i$ from Shannon's equation, and isolate $C$.
Let us suppose then, that $Q$ is square and invertible, with inverse $F = f_{i,j}$.
We have
\begin{equation}
\left (
\begin{array}{c}
\log^* \left ( \sum_i q_{i,1} p_i \right ) \\
\vdots \\
\log^* \left ( \sum_i q_{i,m} p_i \right )
\end{array}
\right )
=
F
\left (
\begin{array}{c}
\sum_j q_{1,j} \log^* q_{1,j} - C \\
\vdots \\
\sum_j q_{n,j} \log^* q_{n,j} - C
\end{array}
\right ).
\end{equation}
Making the simplifying assumption that $p_i, r_j > 0$ for all $i$ and $j$, we exponentiate with base $b$, obtaining
\begin{equation}
\left (
\begin{array}{c}
\sum_i q_{i,1} p_i \\
\vdots \\
\sum_i q_{i,m} p_i
\end{array}
\right )
=
\exp_b \left [
F
\left (
\begin{array}{c}
\sum_j q_{1,j} \log^* q_{1,j} - C \\
\vdots \\
\sum_j q_{n,j} \log^* q_{n,j} - C
\end{array}
\right )
\right ],
\end{equation}
that is,
\begin{equation}
Q^T \vec{p}=
\exp_b \left [
F
\left (
\begin{array}{c}
\sum_j q_{1,j} \log^* q_{1,j} - C \\
\vdots \\
\sum_j q_{n,j} \log^* q_{n,j} - C
\end{array}
\right )
\right ],
\end{equation}
hence
\begin{equation}
\vec{p}=
F^T
\exp_b \left [
F
\left (
\begin{array}{c}
\sum_j q_{1,j} \log^* q_{1,j} - C \\
\vdots \\
\sum_j q_{n,j} \log^* q_{n,j} - C
\end{array}
\right )
\right ].
\end{equation}
So for all $i \in \{1, \ldots, n\}$, we have
\begin{align}
p_i
&=\sum_k (F^T)_{i,k} \exp_b \left [\sum_l f_{k,l} \left (\sum_j q_{l,j} \log^* q_{l,j} - C \right ) \right ] \\
&= \sum_k f_{k,i} \exp_b \left [ -C \sum_l f_{k,l} + \sum_{l,j} f_{k,l} q_{l,j} \log^* q_{l,j} \right ].
\end{align}
Summing over $i$ now yields
\begin{equation}
\label{eq: implicit equation in p and C}
1 = \sum_i p_i = \sum_{i,k} f_{k,i} \exp_b \left [ -C \sum_l f_{k,l} + \sum_{l,j} f_{k,l} q_{l,j} \log^* q_{l,j} \right ].
\end{equation}
Having assumed all $p_i$ to be strictly positive, the row sums of $Q$ all equal $1$, hence the vector $(1, \ldots, 1)^T$ is an eigenvector of $Q$ with eigenvalue $1$, i.e.
\begin{equation}
Q
\left (
\begin{array}{c}
1 \\
\vdots \\
1
\end{array}
\right )
=
\left (
\begin{array}{c}
1 \\
\vdots \\
1
\end{array}
\right ),
\end{equation}
from which
\begin{equation}
\left (
\begin{array}{c}
1 \\
\vdots \\
1
\end{array}
\right )
=
F
\left (
\begin{array}{c}
1 \\
\vdots \\
1
\end{array}
\right ),
\end{equation}
whereby the row sums of $F$, also, are equal to $1$, that is, $\sum_l f_{k,l} = 1$ for all $k$.
Muroga uses a more complicated argument to prove this point, involving cofactors of $Q$.
Notwithstanding those small mistakes we have mentioned, this insight allowed Muroga to isolate $C$ in \eqref{eq: implicit equation in p and C} as follows:
\begin{align}
1
&=\sum_{i,k} f_{k,i} \exp_b \left [ -C \sum_l f_{k,l} + \sum_{l,j} f_{k,l} q_{l,j} \log^* q_{l,j} \right ] \\
&=\sum_{i,k} f_{k,i} \exp_b \left [ -C + \sum_{l,j} f_{k,l} q_{l,j} \log^* q_{l,j} \right ] \\
&=\sum_k \left \{ \exp_b \left [ -C + \sum_{l,j} f_{k,l} q_{l,j} \log^* q_{l,j} \right ] \sum_i f_{k,i} \right \} \\
&=\sum_k \left \{ \exp_b \left [ -C + \sum_{l,j} f_{k,l} q_{l,j} \log^* q_{l,j} \right ] \right \} \\
&=\exp_b(-C) \sum_k \exp_b \left [ \sum_{l,j} f_{k,l} q_{l,j} \log^* q_{l,j} \right ].
\end{align}
Hence
\begin{equation}
C = \log_b \left [\sum_k \exp_b \left ( \sum_{l,j} f_{k,l} q_{l,j} \log^* q_{l,j} \right ) \right ].
\end{equation}
This is essentially equation (8) on pg. 485 of Muroga.

In \cite{Muroga}, Muroga now changes to an approach which is in a sense dual to the preceding argument. He uses the method of Lagrange multipliers to maximise $I(X,Y)$ with respect to $\vec{r}=(r_1, \ldots, r_m) \in \Delta^{m-1}$, instead of with respect to $\vec{p} \in \Delta^{n-1}$. This yields an expression for the optimal $\vec{r}$ which is certainly feasible, i.e. lies in $\Delta^{m-1}$, since the $r_j$ sum to 1, and since they are expressed as exponentials, and are therefore positive. The problem then remains to check whether $\vec{p}$ is also feasible. In \cite{Silverman}, in the special case of a binary channel, Silverman, following Muroga, verifies that $\vec{p}$ is indeed feasible. His method actually yields the stronger conclusion that $p_1, p_2 \in [1/e, 1-1/e]$.
Silverman's argument relies on a careful analysis of $H(X)$ as a function of $p_1$ only, and its first and second derivatives with respect to $p_1$. He leaves several of the details to the reader.
We present here, in case it is of interest, a proof that $p_1$ and $p_2$ must be non-negative, from which it follows, under the assumption that $p_1+p_2=1$, that $\vec{p}$ is feasible.

In Muroga's Theorem 1, in terms of the present article, he begins by supposing that the linear system
\begin{equation}
\label{eq: linear system}
Q
\left (
\begin{array}{c}
X_1 \\
\vdots \\
X_m
\end{array}
\right )
=
\left (
\begin{array}{c}
\sum_j q_{1,j} \log^* q_{1,j} \\
\vdots \\
\sum_j q_{n,j} \log^* q_{n,j}
\end{array}
\right )
\end{equation}
admits solutions $(X_1, \ldots, X_m)^T$.
We will follow Muroga's argument here.
Consider the bilinear form
\begin{equation}
(p_1, \ldots, p_m)
Q
\left (
\begin{array}{c}
X_1 \\
\vdots \\
X_m
\end{array}
\right ).
\end{equation}
It can be expressed as
\begin{equation}
\sum_j X_j \sum_i q_{i,j} p_i = \sum_j X_j \sum_i p_{i,j} = \sum_j X_j r_j,
\end{equation}
or as
\begin{equation}
\sum_i p_i \sum_j q_{i,j} X_j = \sum_i p_i \sum_j q_{i,j} \log^* q_{i,j} = -H(Y|X),
\end{equation}
in view of \eqref{eq: linear system}.
Therefore
\begin{equation}
H(Y|X)=-\sum_j r_j X_j
\end{equation}
and
\begin{align}
I(X,Y)
&=H(Y)-H(Y|X) \\
&=-\sum_j r_j \log^* r_j - H(Y|X) \\
\label{eq: formula for I in terms of the rj}
&= -\sum_j r_j \log^* r_j + \sum_j r_j X_j.
\end{align}
Noting that the $X_j$ depend only on $Q$, our goal is now to maximise this last expression for $I(X,Y)$, with respect to the variables $r_j$.
We assume for simplicity that $\vec{r} \in (\Delta^{m-1})^{\circ}$.
Fixing $l \in \{1, \ldots, m\}$, we have
\begin{align}
\frac{\partial I}{\partial r_l}
&= -\sum_j \frac{\partial}{\partial r_l} (r_j \log_b r_j) + X_l \\
&= -\sum_j \left ( \delta_{j,l} \log_b r_j + r_j \frac{\delta_{j,l}}{r_j} \right ) + X_l \\
&= -(\log_b r_l + 1) + X_l,
\end{align}
so we want
\begin{equation}
\label{eq: the second use of mu}
-\log_b r_l + X_l = \mu \qquad \text{for all } l,
\end{equation}
for some constant $\mu$.
Multiplying by $r_l$ and summing over $l$,
\begin{equation}
\label{eq: apply assumption that the rj sum to 1}
-\sum_l r_l \log_b r_l + \sum_l r_l X_l = \mu \sum_l r_l = \mu.
\end{equation}
In view of \eqref{eq: formula for I in terms of the rj}, it follows that $\mu = C$.
Substituting $\mu=C$ in \eqref{eq: the second use of mu} and isolating $r_l$, we obtain $r_l = \exp_b (X_l - C)$, or
\begin{equation}
\label{eq: rj in terms of Xj and C}
r_j = \exp_b (X_j - C) \qquad \text{for all } j,
\end{equation}
which is Muroga's equation (16).
In particular, $r_j > 0$ for all $j$, and since $\sum_j r_j = 1$ (a condition which we forced in \eqref{eq: apply assumption that the rj sum to 1}), it follows that $\vec{r} \in \Delta^{m-1}$, i.e. that the optimal $\vec{r}$ is in fact feasible.

Summing \eqref{eq: rj in terms of Xj and C} over $j$ we obtain an explicit formula for $C$ as follows:
\begin{equation}
1=\sum_j r_j = \sum_j \exp_b (X_j - C) = \exp_b(-C) \sum_j \exp_b X_j,
\end{equation}
whence
\begin{equation}
\label{eq: C in terms of the Xj}
C = \log_b \left [ \sum_j \exp_b X_j \right ]
\end{equation}

\section{The capacity of a binary channel}

Suppose now that $X$ and $Y$ each have exactly 2 states, and suppose that the transition matrix $Q$ is given by
\begin{equation}
Q=
\left (
\begin{array}{cc}
1-a	&	a \\
1-c	&	c
\end{array}
\right ),
\end{equation}
as in \cite{Moore_and_Shannon}, where $a, c \in [0,1]$.
First suppose that $\vec{p} \in \{(1,0),(0,1)\}=\partial \Delta^1$, the boundary of the standard 2-simplex. Then $H(X) = -1 \cdot \log_b 1 = 0$, and $H(X,Y)=H(Y)$, hence $I(X,Y)=H(X)+H(Y)-H(X,Y)=0$.
Therefore we will suppose from now on that $\vec{p} \in (\Delta^1)^{\circ}$, the interior of the simplex.
Note that $|Q|=c-a$, hence $Q$ is singular if, and only if, $a=c$.
If $a=c$ then, as is pointed out in \cite{Moore_and_Shannon}, the capacity $C$ of the channel is zero.
Indeed, in case $a=c$ we have $H(Y|X)=H(Y)$, hence $I(X,Y) = H(Y)-H(Y|X) = 0$, whatever the choice of $\vec{p} \in \Delta^1$.
Therefore we will assume from now on that $Q$ is nonsingular.

Let us follow Muroga by solving the linear system \eqref{eq: linear system}.
For convenience of notation, we define $a'=1-a$ and $c'=1-c$.
In our case, \eqref{eq: linear system} becomes
\begin{equation}
\left (
\begin{array}{cc}
a'	&	a \\
c'	&	c
\end{array}
\right )
\left (
\begin{array}{c}
X_1 \\
X_2
\end{array}
\right )
=
\left (
\begin{array}{c}
a' \log_2 a' + a \log_2 a \\
c' \log_2 c' + c \log_2 c
\end{array}
\right ).
\end{equation}
The solution to which is
\begin{equation}
\label{eq: solution for X1 and X2}
\left \{
\begin{array}{ll}
X_1
&= \frac{ca'}{c-a} \log_2 a' + \frac{ac}{c-a} \log_2 a - \frac{ac'}{c-a} \log_2 c' - \frac{ac}{c-a} \log_2 c \\
X_2
&= \frac{a'c'}{c-a} \log_2 c' + \frac{a'c}{c-a} \log_2 c - \frac{a'c'}{c-a} \log_2 a' - \frac{ac'}{c-a} \log_2 a.
\end{array}
\right .
\end{equation}
Substituting these values in \eqref{eq: C in terms of the Xj} yields
\begin{align}
C
&=\log_2 \left [ \exp_2 \left (\frac{ca'}{c-a} \log_2 a' + \frac{ac}{c-a} \log_2 a - \frac{ac'}{c-a} \log_2 c' - \frac{ac}{c-a} \log_2 c \right ) \right . \\
&\phantom{shsdofisdfd} \left . + \exp_2 \left ( \frac{a'c'}{c-a} \log_2 c' + \frac{a'c}{c-a} \log_2 c - \frac{a'c'}{c-a} \log_2 a' - \frac{ac'}{c-a} \log_2 a \right ) \right ]
\end{align}
hence
\begin{equation}
\label{eq: C as a function of a and c}
C=\log_2 \left [ \left ( \frac{a'^{a'c}a^{ac}}{c'^{ac'}c^{ac}} \right )^{\frac{1}{c-a}}
+ \left ( \frac{c'^{a'c'}c^{a'c}}{a'^{a'c'}a^{ac'}} \right)^{\frac{1}{c-a}}  \right ].
\end{equation}
This formula is essentially equation (5) in \cite{Silverman}, given the substitution $H(x) = -x \log_2 x - (1-x) \log_2 (1-x)$, and after transforming the right hand side into a symmetric form.

It remains to compute the optimal $p_i$ and to prove that they are feasible.
Denoting the inverse of
$Q = \left (
\begin{array}{cc}
a'	& a \\
c'	& c
\end{array}
\right )
$
by $F$ as in the previous section, we have
\begin{equation}
F=
\frac{1}{c-a}
\left (
\begin{array}{cc}
c	& -a \\
-c'	& a'
\end{array}
\right ).
\end{equation}
For a general square, nonsingular $Q$ with inverse $F$, we have $Q^T (p_1, \ldots, p_n)^T = (r_1, \ldots, r_m)^T$, hence $(p_1, \ldots, p_n)^T = F^T (r_1, \ldots, r_m)^T$.
In the present situation, this means that
\begin{equation}
\left (
\begin{array}{c}
p_1 \\
p_2
\end{array}
\right )
=
\frac{1}{c-a}
\left (
\begin{array}{cc}
c	& -c' \\
-a	& a'
\end{array}
\right ) 
\left (
\begin{array}{c}
r_1 \\
r_2
\end{array}
\right ).
\end{equation}
Combining this last relation with \eqref{eq: rj in terms of Xj and C} and \eqref{eq: C in terms of the Xj} gives
\begin{equation}
p_1=\frac{1}{c-a} \frac{c \exp_2 (X_1) - c' \exp_2 (X_2)}{\exp_2 (X_1) + \exp_2 (X_2)}
\end{equation}
and
\begin{equation}
p_2=\frac{1}{c-a} \frac{-a \exp_2 (X_1) + a' \exp_2(X_2)}{\exp_2 (X_1) + \exp_2 (X_2)},
\end{equation}
from which it follows that $p_1 + p_2 = 1$.
To show the feasibility of $\vec{p}$ it will now suffice to show that $p_1$ and $p_2$ are both non-negative.
In fact, assuming that $0 < a < c < 1$, we will have
$p_1 \geq 0$ iff
\begin{equation}
c \exp_2(X_1) - c' \exp_2(X_2) \geq 0
\end{equation}
iff
\begin{equation}
\frac{c}{c'} \geq \frac{\exp_2(X_2)}{\exp_2(X_1)} = \exp_2(X_2-X_1)
\end{equation}
iff
\begin{equation}
\log_2 c - \log_2 c' \geq X_2 - X_1
\end{equation}
iff
\begin{align}
\log_2 c - \log_2 c'
&\geq \frac{1}{c-a}
\left [
(-ac' - ac)\log_2 a + (a'c+ac)\log_2 c \right . \\
& \phantom{ougaasfrg} \left .+ (-a'c'-ca')\log_2 a' + (a'c'+ac')\log_2 c'
\right ]
\end{align}
by \eqref{eq: solution for X1 and X2},
iff
\begin{equation}
\log_2 c - \log_2 c' \geq \frac{1}{c-a}
\left (
-a \log_2 a + c \log_2 c - a' \log_2 a' + c' \log_2 c'
\right )
\end{equation}
iff
\begin{equation}
(c-a)(\log_2 c - \log_2 c') \geq
-a \log_2 a + c \log_2 c - a' \log_2 a' + c' \log_2 c'
\end{equation}
iff
\begin{equation}
a' \log_2 a' - (c-a+c') \log_2 c' \geq (c-(c-a)) \log_2 c - a \log_2 a
\end{equation}
iff
\begin{equation}
a'(\log_2 a' - \log_2 c') \geq a(\log_2 c - \log_2 a).
\end{equation}
However, considering the graphs of $y=\log_2 x$ and $y = \exp_2 x$, we have
\begin{equation}
a'(\log_2 a' - \log_2 c') \geq \int_{\log_2 c'}^{\log_2 a'} \exp_2 x \, dx = \frac{1}{\ln 2}(a'-c')=\frac{c-a}{\ln 2},
\end{equation}
and similarly,
\begin{equation}
a(\log_2 c - \log_2 a) \leq \int_{\log_2 a}^{\log_2 c} \exp_2 x \, dx = \frac{c-a}{\ln 2}.
\end{equation}
It follows that 
\begin{equation}
a'(\log_2 a' - \log_2 c') \geq a(\log_2 c - \log_2 a),
\end{equation}
and therefore that $p_1 \geq 0$, as required.
An exactly similar method shows that $p_2 \geq 0$, as is also required.

\section{Acknowledgements}
The author would like to thank his colleagues Valeriu Beiu, Leonard D\u{a}u\c{s} and Philippe Poulin for several helpful comments which improved this article.

\end{document}